\pdfoutput=1
\documentclass[1p,number,a4paper]{elsarticle}
\usepackage{amsfonts}
\usepackage{amsmath}
\usepackage{amssymb}
\usepackage{amsthm}
\usepackage[english]{babel}
\usepackage{complexity}
\usepackage{enumerate}
\usepackage{mathrsfs}
\usepackage{graphics}
\usepackage[pagewise]{lineno}
\usepackage{mathtools}
\usepackage{microtype}
\usepackage[defaultlines=4,all]{nowidow}
\usepackage[many]{tcolorbox}
\usepackage{todonotes}
\usepackage{xcolor}
\usepackage{xspace}

\usepackage[hidelinks]{hyperref}
\usepackage[nameinlink]{cleveref}

\usetikzlibrary{arrows, automata, fit, shapes, arrows.meta, positioning}

\tcolorboxenvironment{theorem}{
    colframe=cyan, sharp corners=west, 
    colback=cyan!10,
	rightrule=0pt,
	toprule=0pt,
	bottomrule=0pt,
	top=0pt,
	right=0pt,
	bottom=0pt,
 }
\tcolorboxenvironment{corollary}{
    colframe=blue, sharp corners=west, 
    colback=blue!10,
	rightrule=0pt,
	toprule=0pt,
	bottomrule=0pt,
	top=0pt,
	right=0pt,
	bottom=0pt,
 }
\tcolorboxenvironment{lemma}{
    colframe=red, sharp corners=west,
	colback=red!10,
	rightrule=0pt,
	toprule=0pt,
	bottomrule=0pt,
	top=0pt,
	right=0pt,
	bottom=0pt,
 }

\renewcommand{\epsilon}{\varepsilon}
\renewcommand{\phi}{\varphi}

\newcommand{\ie}{i.\,e.\xspace}
\newcommand{\eg}{e.\,g.\xspace}

\newcommand{\Amc}{\ensuremath{\mathcal{A}}\xspace}
\newcommand{\Bmc}{\ensuremath{\mathcal{B}}\xspace}

\newcommand{\Fmc}{\ensuremath{\mathcal{F}}\xspace}

\newcommand{\Kmc}{\ensuremath{\mathcal{K}}\xspace}
\newcommand{\Lmc}{\ensuremath{\mathcal{L}}\xspace}

\newcommand{\Tmc}{\ensuremath{\mathcal{T}}\xspace}

\newcommand{\Nbb}{\ensuremath{\mathbb{N}}\xspace}

\newcommand{\ebf}{\ensuremath{\mathbf{e}}\xspace}
\newcommand{\ubf}{\ensuremath{\mathbf{u}}\xspace}
\newcommand{\vbf}{\ensuremath{\mathbf{v}}\xspace}

\newcommand{\0}{\ensuremath{\mathbf{0}}\xspace}
\newcommand{\Inf}{\ensuremath{\mathsf{Inf}}}

\newcommand{\LPPBA}{\ensuremath{\Lmc_\mathsf{PPBA}}\xspace}
\newcommand{\LSPBA}{\ensuremath{\Lmc_\mathsf{SPBA}}\xspace}
\newcommand{\LRPA}{\ensuremath{\Lmc_\mathsf{RPA}}\xspace}
\newcommand{\LRPBA}{\ensuremath{\Lmc_\mathsf{RPBA}}\xspace}
\newcommand{\LLPBA}{\ensuremath{\Lmc_\mathsf{LPBA}}\xspace}

\newcommand{\LPAPA}{\ensuremath{\Lmc_\mathsf{PA, PA}^\omega}\xspace}
\newcommand{\LPAReg}{\ensuremath{\Lmc_\mathsf{PA, Reg}^\omega}\xspace}
\newcommand{\LRegPA}{\ensuremath{\Lmc_\mathsf{Reg, PA}^\omega}\xspace}
\newcommand{\LRegReg}{\ensuremath{\Lmc_\mathsf{Reg, Reg}^\omega}\xspace}

\theoremstyle{remark}
\newtheorem{theorem}{Theorem}[section]
\newtheorem{corollary}{Corollary}[section]

\newtheorem{lemma}{Lemma}[section]

\newtheorem{remark}{Remark}[section]
\newtheorem{observation}{Observation}[section]

\newtheorem{example}{Schnexample}
\crefname{corollary}{Corollary}{Corollaries}
\crefname{lemma}{Lemma}{Lemmas}
\crefname{section}{Section}{Sections}

\journal{arXiv}
\begin{document}
\begin{frontmatter}

\title{Büchi-like characterizations for\\ Parikh-recognizable omega-languages}
\author{Mario Grobler}
\ead{grobler@uni-bremen.de}
\address{University of Bremen, Bremen, Germany}
\author{Sebastian Siebertz}\address{University of Bremen, Bremen, Germany}\ead{siebertz@uni-bremen.de}

\begin{abstract}
Büchi's theorem states that $\omega$-regular languages are characterized as languages of the form $\bigcup_i U_i V_i^\omega$, where $U_i$ and $V_i$ are regular languages.
Parikh automata are automata on finite words whose transitions are equipped with vectors of positive integers, whose sum can be tested for membership in a given semi-linear set.
%Just recently, several models for Parikh automata on infinite words were introduced.
We give an intuitive automata theoretic characterization of languages of the form $U_i V_i^\omega$, where~$U_i$ and $V_i$ are Parikh-recognizable.
Furthermore, we show that the class of such languages, where~$U_i$ is Parikh-recognizable and $V_i$ is regular is exactly captured by a model proposed by Klaedtke and Ruess [Automata, Languages and Programming, 2003], which again is equivalent to (a small modification of) reachability Parikh automata introduced by Guha et al.~[FSTTCS, 2022].
We finish this study by introducing a model that captures exactly such languages for regular $U_i$ and Parikh-recognizable $V_i$.
\end{abstract}
\end{frontmatter}

\section{Introduction}

In his groundbreaking work~\cite{buechi} from 1960 Büchi initiated the study of $\omega$-regular languages and introduced Büchi automata. 
By his famous theorem $\omega$-regular languages are characterized as languages of the form $\bigcup_i U_i V_i^\omega$, where the $U_i$ and $V_i$ are regular languages. 
One shortcoming of Büchi automata, and also of many more powerful models, is that they cannot count. 
For example, the language $\{a^nb^nc^n \mid n \in \mathbb{N}\}^\omega$ is not $\omega$-regular, and not even $\omega$-context-free.
This shortcoming led to the study of automata on infinite words with counters, see e.g.~\cite{allredCounting, bojanczykbeyond, beyondomegabs}. 

Parikh automata (PA) are another model of automata (on finite words) with counters~\cite{klaedtkeruess}. 
A PA with~$d$~counters is a non-deterministic finite automaton that is additionally equipped with a semi-linear set~$C$. 
Furthermore, every transition is equipped with a $d$-tuple of non-negative integers and every time a transition is used, the counters are incremented by the values in the tuple accordingly. 
A finite input word is accepted if the PA ends in a final state and additionally, the resulting $d$-tuple in the counters lies in~$C$. The class of languages recognized by PA contains all regular languages, and even some, but not all, context-sensitive languages, e.g.\ the language $\{a^nb^nc^n \mid n \in \mathbb{N}\}$.

Recently, several possible extensions of Parikh automata on infinite words were proposed and studied by Grobler et al.~\cite{infiniteOurs} and Guha et al.~\cite{infiniteZimmermann}. 
In fact, it turns out that one of the models proposed in~\cite{infiniteOurs} is equivalent to synchronous blind counter machines, which were introduced by Fernau and Stiebe~\cite{blindcounter}. 
Fernau and Stiebe also considered the class of all $\omega$-languages of the form $\bigcup_iU_iV_i^\omega$, where the $U_i,V_i$ are Parikh-recognizable languages of finite words. 
They called this class $\Kmc_*$ and proved that the class of \mbox{$\omega$-languages} recognized by blind counter machines is a proper subset of $\Kmc_*$. 

\smallskip
In the light of Büchi's famous theorem it is a natural question to find an automata theoretic characterization of the class $\Kmc_*$. In fact, more generally, we define the four classes $\LRegReg$, $\LPAReg$, $\LRegPA$ and $\LPAPA$ of $\omega$-languages of the form $\bigcup_iU_iV_i^\omega$, where the $U_i,V_i$ are regular or Parikh-recognizable languages of finite words, respectively. By Büchi's theorem the class $\LRegReg$ is the class of $\omega$-regular languages. 
In this work we provide automata theoretic characterizations of the other three classes. 

\smallskip
We first introduce the new model of \emph{limit Parikh Büchi automata} (LPBA), which was suggested in the concluding remarks of Klaedtke and Ruess~\cite{klaedtkeruess}. 
An LPBA accepts an infinite word if an accepting state is visited infinitely often (satisfies the Büchi condition) and the infinite sum of the counters belongs to the semi-linear set, which for this purpose is extended with the symbol $\infty$ if the sum of some counter diverges (satisfies the newly introduced \emph{limit Parikh condition}). 

We also introduce a new model, which is obtained by a small modification of reachability Parikh automata as introduced by Guha et al.~\cite{infiniteZimmermann}, that we call \emph{reachability Parikh Büchi automata} 
(RPBA). An RPBA accepts an infinite word if an 
accepting state is visited infinitely often (satisfies the Büchi condition) and satisfies the Parikh condition \emph{once}. 

Quite surprisingly, both models turn out to capture exactly the class $\LPAReg$, and hence are equivalent. 

\smallskip
We then study \emph{strong reset Parikh automata} (SPBA), which were introduced by Grobler et al.~\cite{infiniteOurs}. We consider the automata as directed graphs and provide two graph theoretic definitions of subclasses of SPBA that exactly capture the classes $\LRegPA$ and~$\LPAPA$. These definitions are based on an analysis of the strongly connected components of the underlying graph, where the accepting states can be found and how they are connected to the rest of the graph. 

\smallskip
We believe that our results provide interesting insights into the theory of Parikh-recognizable $\omega$-languages. 
It remains an interesting open question to characterize the new classes of $\omega$-languages by logics. 

\section{Preliminaries}
\label{sec:prelims}

\subsection{Finite and infinite words}
We write $\Nbb$ for the set of non-negative integers including $0$. Let $\Sigma$ be an alphabet, \ie, a finite non-empty set and let $\Sigma^*$ be the set of all finite words over $\Sigma$. 
For a word $w \in \Sigma^*$, we denote by $|w|$ the length of $w$, and by $|w|_a$ the number of occurrences of the letter $a \in \Sigma$ in $w$. 
We write $\varepsilon$ for the empty word of length~$0$.

An \emph{infinite word} over an alphabet $\Sigma$ is a function $\alpha : \Nbb \setminus \{0\} \rightarrow \Sigma$. We often write~$\alpha_i$ instead of~$\alpha(i)$. 
Thus, we can understand an infinite word as an infinite sequence of symbols $\alpha = \alpha_1\alpha_2\alpha_3\ldots$ For $m \leq n$, we abbreviate the finite infix $\alpha_m \ldots \alpha_n$ by $\alpha[m,n]$. 
%Let $w$ be a (finite) word of length~$n$. Then $w^\omega$ denotes the infinite concatenation of $w$, that is, $w^\omega = \alpha$ where $\alpha[i \cdot n + 1, (i+1) \cdot n] = w$ for all $i \geq 0$. 
We denote by $\Sigma^\omega$ the set of all infinite words over $\Sigma$. 
We call a subset $L \subseteq \Sigma^\omega$ an \emph{$\omega$-language}. 
%Since we mainly deal with $\omega$-languages, we will simply speak of languages.
Moreover, for $L \subseteq \Sigma^*$, we define $L^\omega = \{w_1w_2\dots \mid w_i \in L \setminus \{\varepsilon\}\} \subseteq \Sigma^\omega$.

\subsection{Regular and $\omega$-regular languages}

A \emph{Nondeterministic Finite Automaton} (NFA) is a tuple $\Amc = (Q, \Sigma, q_0, \Delta, F)$, where~$Q$ is the finite set of states, $\Sigma$ is the input alphabet, $q_0 \in Q$ is the initial state, $\Delta \subseteq Q \times \Sigma \times Q$ is the set of transitions and $F \subseteq Q$ is the set of accepting states. 
A \emph{run} of $\Amc$ on a word $w = w_1 \ldots w_n\in \Sigma^*$ is a (possibly empty) sequence of transitions $r = r_1 \ldots r_n$ with $r_i = (p_{i-1}, w_i, p_i)\in \Delta$ such that $p_0=q_0$. 
We say $r$ is \emph{accepting} if $p_n \in F$. The empty run on~$\epsilon$ is accepting if $q_0 \in F$. We define the \emph{language recognized by \Amc} as $L(\Amc) = \{w \in \Sigma^* \mid \text{there is an accepting run of $\Amc$ on $w$}\}$. If a language $L$ is recognized by some NFA $\Amc$, we call $L$ \emph{regular}.

A \emph{Büchi Automaton} (BA) is an NFA $\Amc = (Q, \Sigma, q_0, \Delta, F)$ that takes infinite words as input. 
A \emph{run} of $\Amc$ on an infinite word $\alpha_1\alpha_2\alpha_3\dots$ is an infinite sequence of transitions $r = r_1 r_2 r_3 \dots$ with $r_i = (p_{i-1}, \alpha_i, p_i) \in \Delta$ such that $p_0=q_0$. 
We say $r$ is \emph{accepting} if there are infinitely many~$i$ such that $p_i \in F$. 
We define the \emph{$\omega$-language recognized by~$\Amc$} as $L_\omega(\Amc) = \{\alpha \in \Sigma^\omega \mid \text{there is an accepting run of $\Amc$ on $\alpha$}\}$. 
If an $\omega$-language $L$ is recognized by some BA $\Amc$, we call $L$ \emph{$\omega$-regular}. Büchi's theorem establishes an important connection between regular and $\omega$-regular languages:
\begin{theorem}[Büchi]
\label{thm:buechi}
A language $L \subseteq \Sigma^\omega$ is $\omega$-regular if and only if there are regular languages $U_1, V_1, \dots, U_n, V_n \subseteq \Sigma^*$ for some $n \geq 1$ such that $L = U_1V_1^\omega \cup \dots \cup U_nV_n^\omega$.
\end{theorem}

%\subsection{Labeled alphabets and words}
%
%For $d \geq 1$ we call an alphabet $\Gamma \subseteq \Sigma \times \Nbb^d$ a \emph{labeled alphabet} and refer to the entries of a vector in $\Nbb^d$ as \emph{counters}. We define the \emph{projection} $\chi : \Gamma^* \rightarrow \Sigma^*$ as $\chi(\varepsilon) = \varepsilon,\ \chi(x, \vbf) = x$ for $(x,\vbf)\in \Gamma$, and by induction, $\chi(uw) = \chi(u) \cdot \chi(w)$ for $u\in \Gamma$ and $w\in \Gamma^*$. This definition naturally extends to $\omega$-words by transfinite induction. We define the \emph{extended Parikh image} $\rho : \Gamma^*\rightarrow \Nbb^d$ as 
%$\rho(\varepsilon) = \0,\ \rho(x, \vbf) = \vbf$ for $(x,\vbf)\in \Gamma$, and $\rho(uw) = \rho(u) + \rho(w)$ for $u\in \Gamma$ and $w\in \Gamma^*$. We do not extend this latter definition to $\omega$-words. For a finite or infinite word $v$ over $\Sigma$ we call any finite or infinite word $w$ over $\Gamma$ with $\chi(w)=v$ a \emph{labeled extension of $v$ over $\Gamma$}.

\subsection{Semi-linear sets}
A \emph{linear set} of dimension $d$ for $d \geq 1$ is a set of the form $\{b_0 + b_1z_1 + \dots + b_\ell z_\ell \mid z_1, \dots, z_\ell \in \Nbb\} \subseteq \Nbb^d$ for $b_0,\ldots, b_\ell\in \Nbb^d$.
A \emph{semi-linear set} is the finite union of linear sets.
For vectors $\ubf = (u_1, \dots, u_c)\in \Nbb^c, \vbf = (v_1, \dots, v_d) \in \Nbb^d$, we denote by $\ubf \cdot \vbf = (u_1, \dots, u_c, v_1, \dots, v_d) \in \Nbb^{c+d}$ the \emph{concatenation of $\ubf$ and $\vbf$}. 
We extend this definition to sets of vectors. Let $C \subseteq \Nbb^c$ and $D \subseteq \Nbb^d$. Then $C \cdot D = \{\ubf \cdot \vbf \mid \ubf \in C, \vbf \in D\} \subseteq \Nbb^{c+d}$.
We denote by~$\0^d$ (or simply $\0$ if $d$ is clear from the context) the all-zero vector, and by~$\ebf^d_i$ (or simply $\ebf_i)$ the $d$-dimensional vector where the $i$th entry is~$1$ and all other entries are~0.
We also consider semi-linear sets over $(\Nbb \cup \{\infty\})^d$, that is semi-linear sets with an additional symbol $\infty$ for infinity. As usual, addition of vectors and multiplication of a vector with a number is defined component-wise, where $z + \infty = \infty + z = \infty + \infty = \infty$ for all $z \in \Nbb$, $z \cdot \infty = \infty \cdot z = \infty$ for all $z > 0\in \Nbb$, and $0 \cdot \infty = \infty \cdot 0 = 0$.

\subsection{Parikh-recognizable languages}
A \emph{Parikh Automaton} (PA) is a tuple $\Amc = (Q, \Sigma, q_0, \Delta, F, C)$ where $Q$, $\Sigma$, $q_0$, and~$F$ are defined as for NFA, $\Delta \subseteq Q \times \Sigma \times \Nbb^d \times Q$ is the set of \emph{labeled transitions}, and $C \subseteq \Nbb^d$ is a semi-linear set. 
We call $d$ the \emph{dimension} of $\Amc$ and refer to the entries of a vector $\vbf$ in a transition $(p, a, \vbf, q)$ as \emph{counters}.
%Note that $C$ can equivalently be represented by a family of polynomials of the form $g(z_1, \dots, z_\ell) = b_0 + \sum_{i=1}^\ell b_iz_i$ or a Presburger formula $\varphi$ with $d$ free variables (we refer to \cite{haase} for more details). 
Similar to NFA, a \emph{run} of $\Amc$ on a word $w = x_1 \dots x_n$ is a (possibly empty) sequence of labeled transitions $r = r_1 \dots r_n$ with $r_i = (p_{i-1}, x_i, \vbf_i, p_i) \in \Delta$ such that $p_0 = q_0$. We define the \emph{extended Parikh image} of a run $r$ as $\rho(r) = \sum_{i \leq n} \vbf_i$ (with the convention that the empty sum equals $\0$).
%A PA accepts all finite words $v\in \Sigma^*$ such that there exists a labeled extension $w\in \Gamma^*$ of $v$ such that $w$ is accepted by $\Amc$ and $\rho(w)\in C$, \ie, $L(\Amc, C) = \{v\in \Sigma^* \mid \text{there exists $w \in L(\Amc)$ with $\chi(w)=v$ and $\rho(w) \in C\}$}$.
We say $r$ is accepting if $p_n \in F$ and $\rho(r) \in C$,
referring to the latter condition as the \emph{Parikh condition}. 
We define the \emph{language recognized by \Amc} as $L(\Amc) = \{w \in \Sigma^* \mid \text{there is an accepting run of $\Amc$ on $w$}\}$. 
If a language $L\subseteq \Sigma^*$ is recognized by some PA, then we call $L$ \emph{Parikh-recognizable}.

\subsection{Graphs}
A \emph{(directed) graph} $G$ consists of its vertex set $V(G)$ and edge set \mbox{$E(G) \subseteq V(G) \times V(G)$}. In particular, a graph $G$ may have loops, that is, edges of the form $(u, u)$. A \emph{path} from a vertex $u$ to a vertex $v$ in $G$ is a sequence of pairwise distinct vertices $v_1 \dots v_k$ such that $v_1 = u$, $v_k = v$, and $(v_i, v_{i+1}) \in E(G)$ for all $1 \leq i < k$. Similarly, a \emph{cycle} in $G$ is a sequence of pairwise distinct vertices $v_1 \dots v_k$ such that $(v_i, v_{i+1}) \in E(G)$ for all $1 \leq i < k$, and $(v_k, v_1) \in E(G)$. If $G$ has no cylces, we call $G$ a directed acyclic graph (DAG).
For a subset $U \subseteq V(G)$, we denote by $G[U]$ the graph $G$ \emph{induced by} $U$, \ie, the graph with vertex set $U$ and edge set $\{(u,v) \in E(G) \mid u, v \in U\}$.
A \emph{strongly connected component} (SCC) in $G$ is a maximal subset $U \in V(G)$ such that for all $u, v \in U$ there is a path from~$u$ to $v$, \ie, all vertices in $U$ are reachable from each other. 
We write $SCC(G)$ for the set of all strongly connected components of $G$ (observe that $SCC(G)$ partitions~$V(G)$). 
The \emph{condensation} of $G$, written $C(G)$, is the DAG obtained from $G$ by contracting each SCC of $G$ into a single vertex, that is $V(C(G)) = SCC(G)$ and $(U, V) \in E(C(G))$ if and only if there is $u \in U$ and $v \in V$ with $(u, v) \in E(G)$. 
%Abusing the notion of DAGs, 
We call the SCCs with no outgoing edges in $C(G)$ leaves.
Note that an automaton can be seen as a labeled graph. Hence, all definitions translate to automata by considering the underlying graph (to be precise, an automaton can be seen as a labeled multigraph; however, we simply drop parallel edges).

\section{Parikh automata on infinite words}
In this section, we recall the relevant definitions of Parikh automata operating on infinite words introduced by Grobler et al. \cite{infiniteOurs} and Guha et al. \cite{infiniteZimmermann}. We then propose further definitions and compare the resulting automata.

\smallskip
A \emph{Parikh-Büchi automaton} (PBA) is a PA $\Amc = (Q, \Sigma, q_0, \Delta, F, C)$. A run of $\Amc$ on an infinite word $\alpha = \alpha_1 \alpha_2 \alpha_3 \dots$ is an infinite sequence of labeled transitions $r = r_1 r_2 r_3 \dots$ with $r_i = (p_{i-1}, \alpha_i, \vbf_i, p_i)\in \Delta$ such that $p_0 = q_0$. The automata defined below differ only in their acceptance conditions. In the following, whenever we say that an automaton $\Amc$ accepts an infinite word $\alpha$, we mean that there is an accepting run of $\Amc$ on $\alpha$.

Let us first recall the definition of \emph{(strong) reset PBA} (SPBA) introduced by Grobler et al.~\cite{infiniteOurs}: let $k_0 = 0$ and denote by $k_1, k_2, \dots$ the positions of all accepting states in~$r$. Then $r$ is accepting if $k_1, k_2, \dots$ is an infinite sequence and $\rho(r_{k_{i-1}+1} \dots r_{k_i}) \in C$ for all $i \geq 1$. The $\omega$-language recognized by an SPBA $\Amc$ is $S_\omega(\Amc) = \{\alpha \mid \Amc \text{ accepts }\alpha\}$. Intuitively worded, whenever an SPBA enters an accepting state, the Parikh condition \emph{must} be satisfied. Then the counters are reset.

Let us now recall the definition of \emph{(synchronous) reachability Parikh automata} (RPA) introduced by Guhe et al.~\cite{infiniteZimmermann}. The run $r$ is accepting if there is an $i \geq 1$ such that $p_i \in F$ and $\rho(r_1 \dots r_i) \in C$. We say there is an accepting hit in $r_i$. The $\omega$-language recognized by an RPA $\Amc$ is $R_\omega(\Amc) = \{\alpha \mid \Amc \text{ accepts }\alpha\}$.

Let us finally recall the definition of \emph{prefix PBA} (PPBA) introduced by Grobler et al.~\cite{infiniteOurs}, which are obviously equivalent to \emph{(synchronous) Büchi Parikh automata} introduced by Guhe et al. \cite{infiniteZimmermann}. The run $r$ is accepting if there are infinitely many $i \geq 1$ such that $p_i \in F$ and $\rho(r_1 \dots r_i) \in C$. The $\omega$-language recognized by a PPBA $\Amc$ is $P_\omega(\Amc) = \{\alpha \mid \Amc \text{ accepts }\alpha\}$. Hence, a PPBA can be seen as a stronger variant of RPA where we require infinitely many accepting hits instead of a single one.

We remark that we stick to the notation of the original papers, that is, we abbreviate (strong) reset Parikh-Büchi automata by SPBA and (synchronous) reachability Parikh automata by RPA. 
The attentive reader may have noticed that we do not use the term~RPBA. 
The reason for this (in addition to sticking to the original notation) is that, unlike Büchi automata, RPA do \emph{not} need to see an accepting state infinitely often. 
We show that this property implies that there are $\omega$-regular languages that are not RPA-recogniazble. 
This motivates the study of RPA that additionally need to satisfy the Büchi-condition (which we hence will call RPBA). 

\smallskip
We begin with a simple lemma. A similar result was proved in Theorem~3 of~\cite{infiniteZimmermann}, however, RPA in~\cite{infiniteZimmermann} are assumed to be \emph{complete}, i.e., for every state and every symbol there is at least one transition. The proof presented in~\cite{infiniteZimmermann} does not go through for the more general setting of non-complete RPA. 

\begin{lemma}
\label{lem:RPAregular}
 There is an $\omega$-regular language that is not recognized by any RPA.
\end{lemma}
\begin{proof}
    We show that $L = \{\alpha \in \{a,b\}^\omega \mid |\alpha|_a = \infty\}$ is not RPA-recognizable.
    Assume that there is an RPA $\Amc$ with $R_\omega(\Amc) = L$ and let $n$ be the number of states of~$\Amc$. As $\alpha = (a^n b^n)^\omega \in L$, there is an accepting run $r = r_1 r_2 r_3 \dots$ of $\Amc$ on $\alpha$. Let $i$ be the first position of an accepting hit in $r$.
    Now consider $\beta = (a^n b^n)^i \cdot b^\omega \notin L$. 
    As $\alpha$ and $\beta$ share the same prefix of length $i$, the run $r_1 \dots r_i$ is also a partial run on~$\beta[1,i]$, hence, generating an accepting hit. 
    Now observe that in every infix of the form $b^n$ a state is visited at least twice by the pigeonhole principle. Hence, we can ``infinitely pump" a $b$-block after the accepting hit in $r_i$ and obtain an accepting run on~$\beta$, contradicting $R_\omega(\Amc) = L$.
\end{proof}
\begin{remark}
\label{remark:RPA}
    In fact, complete RPA are strictly weaker than general RPA, as there is no complete RPA recognizing $\{a\}^\omega$ over $\Sigma = \{a,b\}$; an $\omega$-language that is obviously RPA-recognizable.
\end{remark}

As mentioned above, this weakness motivates the study of \emph{reachability Parikh-Büchi automata} (RPBA). 
%Again, an RPBA is syntactically equivalent to a PA $\Amc = (Q, \Sigma, q_0, \Delta, F, C)$. Let $\alpha = \alpha_1 \alpha_2 \alpha_3$ and $r = r_1 r_2 r_3$ with $r_i = (p_{i-1}, \alpha_i, \vbf_i, p_i)$ be a run of $\Amc$ on $\alpha$. 
The run $r$ is accepting if there is an $i \geq 1$ such that $\rho(r_1 \dots r_i) \in C$ and $p_i \in F$, and if there are infinitely many $j$ such that $p_j \in F$. We define
the $\omega$-language recognized by an RPBA $\Amc$ as $B_\omega(\Amc) = \{\alpha \mid \Amc \text{ accepts }\alpha\}$.

Every $\omega$-regular language is RPBA-recognizable, as we can turn an arbitrary Büchi automaton into an equivalent RPBA by labeling every transition with $0$ and \mbox{setting~$C = \{0\}$}.

\smallskip
Finally, we define a variant of PBA motivated by a proposal of Klaedke and Ruess~\cite{klaedtkeruess}. Here we consider semi-linear sets over $(\Nbb \cup \{\infty\})^d$ and compute the extended Parikh image of an infinite run using transfinite induction.
A \emph{limit Parikh-Büchi automaton} (LPBA) is a PA $\Amc = (Q, \Sigma, q_0, \Delta, F, C)$ where $C$ may use the symbol $\infty$. 
%Let $\alpha = \alpha_1 \alpha_2 \alpha_3\ldots$ and $r = r_1 r_2 r_3\ldots$ with $r_i = (p_{i-1}, \alpha_i, \vbf_i, p_i)$ be a run of $\Amc$ on $\alpha$. 
The run $r$ is accepting
if there are infinitely many~$i \geq 1$ such that $p_i \in F$, and if additionally $\rho(r) \in C$, where the $j$-th component of $\rho(r)$ is computed as follows. If there are infinitely many~$i \geq 1$ such that the $j$-th component of $\vbf_i$ has a non-zero value, then the $j$-th component of $\rho(C)$ is~$\infty$. In other words, if the sum of values in a component diverges, then its value is set to $\infty$.  
Otherwise, the infinite sum yields a positive integer. We define the $\omega$-language recognized by an LPBA $\Amc$ as $L_\omega(\Amc) = \{\alpha \mid \Amc \text{ accepts }\alpha\}$.

\begin{example}
\label{ex}
\begin{figure}
	\centering
	\begin{tikzpicture}[->,>=stealth',shorten >=1pt,auto,node distance=3cm, semithick]
	\tikzstyle{every state}=[minimum size=1cm]
	
	\node[initial, initial text={}, state] (q0) {$q_0$};
	\node[state, accepting] (q1) [right of=q0] {$q_1$};	
	
	\path
	(q0) edge [loop above] node {$a, \begin{pmatrix}1\\0\end{pmatrix}$} (q0)
	(q0) edge              node {$b, \begin{pmatrix}0\\1\end{pmatrix}$} (q1)
	(q1) edge [loop above] node {$b, \begin{pmatrix}0\\1\end{pmatrix}$} (q1)
	(q1) edge [bend left]  node {$a, \begin{pmatrix}1\\0\end{pmatrix}$} (q0)	
	;
	\end{tikzpicture}
	\caption{The automaton $\Amc$ with $C=\{(z,z), (z, \infty) \mid z \in \Nbb\}$ from \Cref{ex}.}
	\label{fig:ex}
\end{figure}

Let $\Amc$ be the automaton in \Cref{fig:ex} with $C = \{(z,z), (z, \infty) \mid z \in \Nbb\}$. 
\begin{itemize}
    \item If we interpret $\Amc$ as a PA (over finite words), then we have $L(\Amc) = \{w \in \{a,b\}^* \cdot \{b\} \mid |w|_a = |w|_b\}$. The automaton is in the accepting state after reading a $b$. The first counter counts the number of read $a$s, the second one reads the number of read $b$s. By definition of $C$ the automaton only accepts when both counters are equal (note that vectors containing an $\infty$-entry have no additional effect). 
    \item If we interpret $\Amc$ as an SPBA, then we have $S_\omega(\Amc) = \{ab\}^\omega$. Whenever the automaton reaches an accepting state also the Parikh condition must be satisfied. In the example this is only possible after reading exactly one $a$ and one $b$. After that the counters are reset. 
    \item If we interpret $\Amc$ as a PPBA, then we have $P_\omega(\Amc) = L(\Amc)^\omega$. The automaton accepts a word if infinitely often the Parikh condition is satisfied in the accepting state. Observe that $C$ has no base vector and the initial state as well as the accepting state have the same outgoing edges.
    \item If we interpret $\Amc$ as an LPBA, then we have $L_\omega(\Amc) = \{\alpha \in \{a,b\}^\omega \mid |\alpha|_a < \infty\}$. The automaton must visit the accepting state infinitely often. At the same time the extended Parikh image must belong to $C$, which implies that the word contains only some finite number $z$ of $a$s (note that only the vectors of the form $(z, \infty)$ have an effect here, as at least one symbol must be seen infinitely often by the infinite pigeonhole principle).
    \item If we interpret $\Amc$ as an RPA, then we have $R_\omega(\Amc) = \{\alpha \in \{a,b\}^\omega \mid \alpha \text{ has a prefix in }$ $L(\Amc)\}$. The automaton has satisfied the reachability condition after reading a prefix in $L(\Amc)$. Since the automaton is complete it cannot get stuck and accepts any continuation after that. 
    \item If we interpret $\Amc$ as an RPBA, then we have $B_\omega(\Amc) = \{\alpha \in \{a,b\}^\omega \mid \alpha \text{ has a prefix}$ in $L(\Amc) \text{ and } |\alpha|_b = \infty\}$. After having met the reachability condition the automaton still needs to satisfy the Büchi condition, which enforces infinitely many visits of the accepting state. 
\end{itemize}
\end{example}

\begin{remark}
The automaton $\Amc$ in the last example is deterministic. We note that $L_\omega(\Amc)$ is not deterministic $\omega$-regular but deterministic LPBA-recognizable.
\end{remark}

We denote the class of $\omega$-languages recognized by an SPBA (PPBA, RPA, RPBA, LPBA) by \LSPBA (\LPPBA, \LRPA, \LRPBA, \LLPBA). Furthermore, denote by $\LPAPA$ the class of $\omega$-languages of the form $\bigcup_i U_iV_i^\omega$, where the $U_i$ and $V_i$ are Parikh-recognizable, by $\LPAReg$ such languages where the $U_i$ are Parikh-recognizable and the $V_i$ are regular, and by $\LRegPA$ such languages where the $U_i$ are regular and the $V_i$ are Parikh-recognizable. 
As shown by Guhe et al.\ we have $\LRPA \subsetneq \LPPBA$ \cite{infiniteZimmermann}. Likewise, Grobler et al.~\cite{infiniteOurs} have shown $\LPPBA \subsetneq \LPAPA \subsetneq \LSPBA$. We conclude this section by showing $\LRPA \subsetneq \LRPBA \subsetneq \LPPBA$. In the next section we show that $\LRPBA = \LLPBA = \LPAReg$.

\begin{lemma}
\label{lem:RPARPBA}
$\LRPA \subsetneq \LRPBA$.
\end{lemma}
\begin{proof}
 We show $\LRPA \subseteq \LRPBA$. Strictness follows from \Cref{lem:RPAregular}.
 The proof is very similar to the proof that $\LRPA \subsetneq \LPPBA$ \cite[Theorem 3]{infiniteZimmermann}.
 Let $\Amc = (Q, \Sigma, q_0, \Delta, F, C)$ be an RPA. The idea is to create two copies of $\Amc$, where the second copy is modified such that all counters are zero and all states are accepting. Then we use non-determinism to guess the accepting hit and transition into the second copy where the counters are frozen and all states are accepting.
 Let $\Amc' = (Q', \Sigma, q_0, \Delta', F', C)$ where $Q' = \{q, q' \mid q \in Q\}$, $F' = \{q' \mid q \in Q\}$ and $ \Delta' = \Delta \cup \{(p',a, \0, q') \mid (p, a, \vbf, q) \in \Delta\} \cup \{(p, a, \vbf, q') \mid (p, a, \vbf, q) \in \Delta, q \in F\}$ be a RPBA.
 We claim that $R_\omega(\Amc) = B_\omega(\Amc')$.

 \smallskip
 $\Rightarrow$ To show $R_\omega(\Amc) \subseteq B_\omega(\Amc')$, let $\alpha \in R_\omega(\Amc)$ with accepting run $r = r_1 r_2 r_3 \dots$ where $r_i = (p_{i-1}, \alpha_i, \vbf_i, p_i)$. Let $i \geq 1$ be an arbitrary position such that there is an accepting hit in $r_i$ (which exists by definition). Let $r'_i = (p_{i-1}, \alpha_i, \0, p_i')$ and define $r'_j = (p'_{j-1}, \alpha_j, \0, p'_j)$ for all $j > i$. 
 Then $r'=r_1 r_2 \dots r_{i-1} r'_i r'_{i+1} r'_{i+2} \dots$ is a run of $\Amc'$ on $\alpha$. 
 Furthermore, $r'$ is accepting: the accepting hit in $r_i$ translates one-to-one to an accepting hit in $r'_i$. Furthermore, all $p_j$ for $j \geq i$ are accepting by the definition of $F'$. Hence, $r'$ is accepting, thus $\alpha \in B_\omega(\Amc')$.

 \smallskip
 $\Leftarrow$ To show $B_\omega(\Amc') \subseteq R_\omega(\Amc)$, let $\alpha \in B_\omega(\Amc')$ with accepting run $r' = r'_1 r'_2 r'_3 \dots$ where $r'_i = (\hat p_{i-1}, \alpha_i, \vbf_i, \hat p_i)$ with $\hat p_i \in \{p_i, p'_i\}$. Again, let $i \geq 1$ be an arbitrary position such that there is an accepting hit in $r'_i$. 
 As there are no accepting states in the first copy of $\Amc$ in $\Amc'$, there is a point where $r$ transitions from the first copy to the second copy, \ie, there is a $j \leq i$ such that $r'_j = (p_{j-1}, \alpha_j, \vbf_j, p_j')$. 
 Observe that the counters are frozen after transitioning to the second copy, hence, we have $\rho(r'_1 \dots r'_i) = \rho(r'_1 \dots r'_j) \in C$. Furthermore, we have $p_j \in F$ by the choice of $\Delta'$. 
 Finally, observe that $\hat p_\ell = p_\ell$ for all $\ell < j$, and $\hat p_\ell = p'_\ell$ for all $\ell \geq j$. 
 Hence, we can replace $r'_j$ by $r_j = (p_{j-1}, \alpha_j, \vbf_j, p_j)$, and for all $\ell > j$ we replace $r'_\ell$ by $r_\ell = (p_{\ell - 1}, \alpha_\ell, \vbf_\ell, p_\ell)$, where $\vbf_\ell$ is arbitrary such that $r_\ell \in \Delta$ (observe that at least one such $\vbf$ exists by definition of $\Delta'$). Then $r = r'_1 r'_2 \dots r'_{j-1} r_j r_{j+1} r_{j+2} \dots$ is a run of $\Amc$ on $\alpha$ that is furthermore accepting as witnessed by the accepting hit in $r_j$. Hence $\alpha \in R_\omega(\Amc)$.
\end{proof}

Observe that a very similar construction can be used to turn an arbitrary RPBA into an equivalent PPBA. The only difference is that we choose $F' = \{q' \mid q \in F\}$. Hence we obtain the following corollary.
\begin{corollary}
\label{cor:LRPBAtoLPPBA}
$\LRPBA \subseteq \LPPBA$.
\end{corollary}

Finally, we show that this inclusion is also strict.
\begin{lemma}
\label{lem:LRPBAstrict}
 There is an $\omega$-language that is PPBA-recognizable but not RPBA-recognizable.
\end{lemma}
\begin{proof}
    Consider $L = \{\alpha \in \{a,b\}^\omega \mid |\alpha[1, i]|_a = |\alpha[1,i]|_b \text{ for infinitely many $i$}\}$, which is obviously PPBA-recognazable. 

    Assume that $L$ is recognized by an RPBA $\Amc$ and let $n$ be the number of states of $\Amc$. Consider an accepting run $r = r_1 r_2 r_3\ldots$ of $\Amc$ on $\alpha = (a^n b^n)^\omega$ where $r_i = (p_{i-1}, \alpha_i, \vbf_i, p_i)$, and let $k$ be the position of the accepting hit, \ie $p_k \in F$ and $\rho(r_1 \dots r_k) \in C$. By definition there are infinitely many $j \geq 1$ (and hence infinitely many $j \geq k$) such that $p_j \in F$. 
    By the pigeonhole principle, there is a state $q$ that is visited twice while reading an arbitrary $a^n$-infix, say at positions $k \leq c < d$, \ie, $p_c = p_d = q$ and $d - c < n$. 
    Hence, we can pump an infix of the form $a^{d-c}$ and obtain an accepting run on an infinite word of the form $(a^n b^n)^* (a^{n + d - c} b^n) (a^n b^n)^\omega$, which is not in~$L$, a contradiction.    
\end{proof}

\section{Characterization of \LPAReg by limit PBA and reachability PBA}
The main goal of this section is to prove the following theorem. 
\begin{theorem}
\label{thm:LimitEqualsReach}
 The following are equivalent for all $\omega$-languages $L \subseteq \Sigma^\omega$.
 \begin{enumerate}
    \item $L$ is of the form $\bigcup_i U_i V_i^\omega$, where $U_i \in \Sigma^*$ is Parikh-recognizable, and $V_i \subseteq \Sigma^*$ is regular.
    \item $L$ is LPBA-recognizable.
    \item $L$ is RPBA-recognizable.
\end{enumerate}
\end{theorem}

Observe that in the first item we may assume that $L$ is of the form $\bigcup_i U_i V_i$, where \mbox{$U_i \in \Sigma^*$} is Parikh-recognizable, and $V_i \subseteq \Sigma^\omega$ is $\omega$-regular. Then, by simple combinatorics and Büchi's theorem we have $\bigcup_i U_i V_i = \bigcup_i U_i (\bigcup_{j_i} X_{j_i} Y_{j_i}^\omega) = \bigcup_{i, j_i} U_i (X_{j_i} Y_{j_i}^\omega) = \bigcup_{i, j_i} (U_i X_{j_i}) Y_{j_i}^\omega$, for regular languages $X_{j_i}, Y_{j_i}$, where $U_i X_{j_i}$ is regular, as regular languages are closed under concatenation.

To simplify the proof, it is convenient to consider the following generalizations of Büchi automata. A \emph{generalized Büchi automaton} (GBA) is a tuple $\Amc = (Q, \Sigma, q_0, \Delta, \Fmc)$ where $Q, \Sigma, q_0$ and $\Delta$ are defined as for BA, and $\Fmc \subseteq 2^Q$ is a collection of sets of accepting states. Then a run $r_1 r_2 r_3 \dots$ with $r_i = (p_{i-1}, \alpha_i, p_i)$ is accepting if for all $F \in \Fmc$ there are infinitely many $i$ such that $p_i \in F$. It is well-known that GBA are not more expressive than BA, see \eg \cite[Theorem 4.56]{modelchecking}. 

Furthermore, we consider a variant of GBA where acceptance is not defined via states that are seen infinitely often, but rather via transitions that are used infinitely often. A \emph{Generalized Transition Büchi Automaton} (GTBA) is a tuple $\Amc = (Q, \Sigma, q_0, \Delta, \Tmc)$ where $\Tmc \subseteq 2^\Delta$ is a collection of sets of transitions. Then a run $r_1 r_2 r_3 \dots$ is accepting if for all $T \in \Tmc$ there are infinitely many $i$ such that $r_i \in T$. 

\begin{lemma}\label{lem:GTBA-BA}
GTBA and BA have the same expressiveness.
\end{lemma}

\begin{proof}
Obviously, every BA can be turned into an equivalent GTBA by choosing $\Tmc = \{ \{(p,a,q) \mid q \in F\}\}$, hence we focus on the other direction.

Let $\Amc = (Q, \Sigma, q_0, \Delta, \Tmc)$ be a GTBA. As GBA are as expressive as BA, it is sufficient to convert $\Amc$ into an equivalent GBA.
The idea is basically to consider the line graph of $\Amc$, that is, to use $\Delta$ as the new state set, keeping the initial state.
Then there is a $b$-transition from the \emph{state} $(p, a, q)$ to every state of the form $(q, b, t)$. Hence, the acceptance component translates directly.
To be precise, we construct the GBA $\Amc' = (\Delta \cup \{q_0\}, \Sigma, q_0, \Delta', \Tmc)$ where $\Delta' = \{((p,a,q), b, (q,b,t) \mid (p,a,q), (q,b,t) \in \Delta\} \cup \{(q_0, a, (q_0, a, q) \mid (q_0, a, q) \in \Delta\}$. It is now easily verified that $L_\omega(\Amc) = L_\omega(\Amc')$.
%$\Rightarrow$ To show $L_\omega(\Amc) \subseteq L_\omega(\Amc')$, let $\alpha \in L_\omega(\Amc)$ with accepting run $r = r_1 r_2 r_3 \dots$, \ie, for all $T \in \Tmc$ there are infinitely many $i$ such that $r_i \in T$.
%Now observe that $r' = (q_0, \alpha_1, r_1) (r_1, \alpha_2, r_2) \dots$ is an accepting run of $\Amc'$ on $\alpha$ as the $r_i$ from the GTBA acceptance translate one-to-one, hence $\alpha \in L_\omega(\Amc')$.
%
%$\Leftarrow$ To show $L_\omega(\Amc') \subseteq L_\omega(\Amc)$, let $\alpha \in L_\omega(\Amc')$ with accepting run $r = (q_0, \alpha_1, r_1) (r_1, \alpha_2, r_2) \dots$, where $r_i \in \Delta$ for all $i \geq 1$. By the definition of acceptance, for every $F \in \Fmc$ there are infinitely many $i$ such that $r_i \in F$.
%Similar to the forward direction, $r_1 r_2 r_3\dots$ is an accepting run of $\Amc$ on $\alpha$, hence $\alpha \in L_\omega(\Amc)$.
\end{proof}

\Cref{thm:LimitEqualsReach} will be a direct consequence from the following lemmas. The first lemma shows the implication $(1) \Rightarrow (2)$.

\begin{lemma}
 If $L \in \LPAReg$, then $L$ is LPBA-recognizable.   
\end{lemma}
\begin{proof}
    First observe that $\LLPBA$ is closed under union (this can be shown using a standard construction).
    Hence, it is sufficient to show how to construct an LPBA for an $\omega$-language of the form $L = UV^\omega$, where $U$ is Parikh-recognizable and $V$ is regular.

    Let $\Amc_1 = (Q_1, \Sigma, q_1, \Delta_1, F_1, C)$ be a PA with $L(\Amc_1) = U$ and $\Amc_2 = (Q_2, \Sigma, q_2, \Delta_2, F_2)$ be a Büchi automaton with $L_\omega(\Amc_2) = V^\omega$. We use the following standard construction for concatenation. Let $\Amc = (Q_1 \cup Q_2, \Sigma, q_1, \Delta, F_2, C)$ be an LPBA where 
    $$\Delta = \Delta_1 \cup \{(p, a, \0, q) \mid (p, a, q) \in \Delta_2\} \cup \{(f, a, \0, q) \mid (q_2, a, q) \in \Delta_2, f \in F_1\}.$$
    We claim that $L_\omega(\Amc) = L$.

    \smallskip
    $\Rightarrow$ To show $L_\omega(\Amc) \subseteq L$, let $\alpha \in L_\omega(\Amc)$ with accepting run $r_1 r_2 r_3 \dots$ where $r_i = (p_{i-1}, \alpha_i, \vbf_i, p_i)$. As only the states in $F_2$ are accepting, there is a position $j$ such that $p_{j-1} \in F_1$ and $p_j \in Q_2$. 
    In particular, all transitions of the copy of $\Amc_2$ are labeled with $\0$, \ie, $\vbf_i = \0$ for all $i \geq j$. Hence $\rho(r) = \rho(r_1 \dots r_{j-1}) \in C$ (in particular, there is no $\infty$ value in $\rho(r)$).
    We observe that $r_1 \dots r_{j-1}$ is an accepting run of $\Amc_1$ on $\alpha[1,j-1]$, as $p_{j-1} \in F_1$ and $\rho(r_1 \dots r_{j-1}) \in C$.
    For all $i \geq j$ let $r'_i = (p_{i-1}, \alpha_i, p_i)$. Now observe that $(q_2, \alpha_j, p_j)r'_{j+1} r'_{j+2} \dots$ is an accepting run of $\Amc_2$ on $\alpha_j \alpha_{j+1} \alpha_{j+2} \dots$, hence $\alpha \in L(\Amc_1) \cdot L_\omega(\Amc_2) = L$.

    \smallskip
    $\Leftarrow$ To show $L = UV^\omega \subseteq L_\omega(\Amc)$, let $w \in L(\Amc_1)=U$ with accepting run $s$, and $\alpha \in L_\omega(\Amc_2)=V^\omega$ with accepting run $r = r_1 r_2 r_3 \dots$, where $r_i = (p_{i-1}, \alpha_1, p_i)$.
    Observe that~$s$ is also a partial run of $\Amc$ on $w$, ending in an accepting state $f$. By definition of~$\Delta$, we can continue the run $s$ in $\Amc$ basically as in $r$. To be precise, let $r'_1 = (f, \alpha_1, \0, p_1)$, and, for all $i > 1$ let $r'_i = (p_{i-1}, \alpha_i, \0, p_i)$. Then $s r'_1 r'_2 r'_3 \dots$ is an accepting run of $\Amc$ on $w \alpha$, hence $w \alpha \in L_\omega(\Amc)$.
\end{proof}

Observe that the construction in the proof of the lemma works the same way when we interpret $\Amc$ as an RPBA (every visit of an accepting state has the same good counter value; this argument is even true if we interpret $\Amc$ as a PPBA), showing the implication~$(1) \Rightarrow (3)$.
\begin{corollary}
If $L \in \LPAReg$, then $L$ is RPBA-recognizable.
\end{corollary}

For the backwards direction we need an auxiliary lemma, essentially stating that semi-linear sets over $C \subseteq (\Nbb \cup \{\infty\})^d$ can be modified such that $\infty$-entries in vectors in~$C$ are replaced by arbitrary integers, and remain semi-linear.

\begin{lemma}
    \label{lem:semi-linear-inf}
    Let $C \subseteq (\Nbb \cup \{\infty\})^d$ be semi-linear and $D \subseteq \{1, \dots, d\}$. Let $C_D \subseteq \Nbb^d$ be the set obtained from $C$ by the following procedure.
    \begin{enumerate}
        \item Remove every vector $\vbf = (v_1, \dots, v_d)$ where $v_i = \infty$ for an $i \notin D$.
        \item As long as $C_D$ contains a vector $\vbf = (v_1, \dots, v_n)$ with $v_i = \infty$ for an $i \leq d$: replace $\vbf$ by all vectors of the form $(v_1, \dots v_{i-1}, z, v_{i+1}, \dots, v_d)$ for $z \in \Nbb$.
    \end{enumerate}
    Then $C_D$ is semi-linear.
\end{lemma}
\begin{proof}
For a vector $\vbf = (v_1, \dots, v_d) \in (\Nbb \cup \{\infty\})^d$, let $\Inf(\vbf) = \{i \mid v_i = \infty\}$ denote the positions of $\infty$-entries in $\vbf$. 
Furthermore, let $\bar\vbf = (\bar{v}_1, \dots, \bar{v}_d)$ denote the vector obtained from $v$ by replacing every $\infty$-entry by 0, \ie, $\bar{v}_i = 0$ if $v_i = \infty$, and $\bar{v}_i = v_i$ otherwise.

We carry out the following procedure for every linear set of the semi-linear set independently, hence we assume that $C = \{b_0 + b_1z_1 + \dots + b_\ell z_\ell \mid z_1, \dots, z_\ell \in \Nbb\}$ is linear. 
We also assume that there is no $b_j$ with $\Inf(b_j) \not\subseteq D$, otherwise, we simply remove it. 

Now, if $\Inf(b_0) \not\subseteq D$, then $C_D = \varnothing$.
Otherwise, $C_D = \{b_0 + \sum_{j\leq \ell} \bar{b}_j z_j + \sum_{i \in \Inf(b_j)} \ebf_i z_{ij} \mid z_j, z_{ij} \in \Nbb\}$, which is linear by definition.
\end{proof}

We are now ready to prove the following lemma, showing the implication $(2) \Rightarrow (1)$.
\begin{lemma}
If $L$ is LPBA-recognizable, then $L \in \LPAReg$. 
\end{lemma}
\begin{proof}
    Let $\Amc = (Q, \Sigma, q_0, \Delta, F, C)$ be an LPBA of dimension $d$. The idea is as follows. We guess a subset $D \subseteq \{1, \dots, d\}$ of counters whose values we expect to be $\infty$.
    Observe that every counter not in $D$ has a finite value, hence for every such counter there is a point where all transitions do not increment the counter further.
    For every subset $D \subseteq \{1, \dots, d\}$ we decompose $\Amc$ into a PA and a GTBA. In the first step we construct a PA where every counter not in $D$ reaches its final value and is verified. In the second step we construct a GTBA ensuring that for every counter in $D$ at least one transition adding a non-zero value to that counter is used infinitely often. This can be encoded directly into the GTBA. Furthermore we delete all transitions that modify counters not in $D$. 

    Fix $D \subseteq \{1, \dots, d\}$ and $f \in F$, and define the PA $\Amc^D_f = (Q, \Sigma, q_0, \Delta, \{f\}, C_D)$ where~$C_D$ is defined as in \Cref{lem:semi-linear-inf}.
    Furthermore, we define the GTBA $\Bmc^D_f = (Q, \Sigma, f, \Delta^D, \Tmc^D)$ where $\Delta^D$ contains the subset of transitions of $\Delta$ where the counters not in $D$ have zero-values (just the transitions without vectors for the counters, as we construct a GTBA). On the other hand, for every counter $i$ in $D$ there is one acceptance component in $\Tmc^D$ that contains exactly those transitions (again without vectors) where the $i$-th counter has a non-zero value. Finally, we encode the condition that at least one accepting state in $F$ needs to by seen in $\Tmc^D$ by further adding the component $\{(p, a, q) \in \Delta \mid q \in F\}$.

    We claim that $L_\omega(\Amc) = \bigcup_{D \subseteq \{1, \dots, d\}, f \in F} L(\Amc^D_f)\cdot L_\omega(\Bmc^D_f)$, which by the comment below \Cref{thm:LimitEqualsReach} and \Cref{lem:GTBA-BA} implies the statement of the lemma. 

    \smallskip
    $\Rightarrow$ To show $L_\omega(\Amc) \subseteq \bigcup_{D \subseteq \{1, \dots, d\}, f \in F} L(\Amc^D_f) \cdot L_\omega(\Bmc^D_f)$, let $\alpha \in L_\omega(\Amc)$ with accepting run $r_1 r_2 r_3 \dots$ where $r_i = (p_{i-1}, \alpha_i, \vbf_i, p_i)$. Let $D$ be the positions of $\infty$-entries in $\rho(r) = (v_1, \dots, v_d)$. As the $v_i$ with $i \notin D$ have integer values, there is a position $j$ such that in all $\vbf_k$ for $k \geq j$ the $i$-th entry of $\vbf_k$ is 0. Let $\ell \geq j$ be minimal such that $p_\ell$ in $F$. We split $\alpha = w \beta$, where $w = \alpha[1,\ell]$, and $\beta = \alpha_{\ell + 1} \alpha_{\ell +2} \dots$.
    
    First we argue that $w \in L_\omega(\Amc^D_{p_{\ell}})$. Observe that $\Amc^D_{p_{\ell}}$ inherits all transitions from $\Amc$, hence $r_1 \dots r_{\ell}$ is a run of $\Amc^D_{p_{\ell}}$ on $w$. As $p_\ell$ is accepting by definition, it remains to show that $\rho(r_1, \dots r_\ell) \in C_D$.
    By the choice of $\ell$, all counters not in $D$ have reached their final values. As $C_D$ contains all vectors of $C$ where all $\infty$-entries are replaced by arbitrary values, the claim follows, hence $w \in L(\Amc^D_{p_{\ell}})$.

    Now we argue that $\beta \in L_\omega(\Bmc^D_{p_\ell})$. For every $k > \ell$ define $r'_k = (p_{k-1}, \alpha_k, p_k)$. Observe that $r' = r'_{k+1} r'_{k+2} \dots$ is a run of $\Bmc^D_{p_\ell}$ on $\beta$ (all $r'_{k+1}$ exist in $\Bmc^D_{p_\ell}$, as the counters not on~$D$ of all transitions $r_k$ have zero-values by the definition of $\ell$).
    It remains to show that~$r'$ is accepting, \ie, that for every counter in $D$ at least one transition with a non-zero values is used infinitely often, and an accepting state is visited infinitely often. This is the case, as these counter values are $\infty$ in $\rho(r)$ and by the acceptance condition of LPBA, hence $\beta \in L_\omega(\Bmc^D_{p_\ell})$.

    We conclude $\alpha \in \bigcup_{D \subseteq \{1, \dots, d\}, f \in F} L(\Amc^D_f) \cdot L_\omega(\Bmc^D_f)$. \hfill $\lrcorner$

\smallskip
    $\Leftarrow$ To show $\bigcup_{D \subseteq \{1, \dots, d\}, f \in F} L(\Amc^D_f) \cdot L_\omega(\Bmc^D_f) \subseteq L_\omega(\Amc)$, let $w \in L(\Amc^D_f)$ and $\beta \in L_\omega(\Bmc^D_f)$ for some $D \subseteq \{1, \dots, d\}$ and $f \in F$. We show that $w\beta \in L_\omega(\Amc)$.

    Let $s$ be an accepting run of $\Amc^D_f$ on $w$, which ends in the accepting state $f$ by definition. Let $\rho(s) = (v_1, \dots, v_d)$. By definition of $C_D$, there is a vector $\ubf = (u_1, \dots, u_d)$ in $C$ where $u_i = \infty$ if $i \in D$, and $u_i = v_i$ if $i \notin D$.
    Furthermore, let $r = r_1r_2r_3\dots$, where $r_i = (p_{i-1}, \alpha_i, p_i)$, be an accepting run of $\Bmc^D_f$ on $\beta$, which starts in the accepting state~$f$ by definition. 
    By definition of $\Tmc^d$, for every counter $i \in D$ at least one transition where the $i$-th counter of the corresponding transition in $\Delta$ is non-zero is used infinitely often. 
    Hence, let $r' = r'_1 r'_2 r'_3 \dots$ where $r'_i = (p_{i-1}, \alpha_i, \vbf_i, p_i)$ for a suitable vector $\vbf_i$. 
    Furthermore, the labels of transitions of counters not in $D$ have a value of zero, hence $\rho(r') = (x_1, \dots, x_d)$, where $x_i = \infty$ if $i \in D$, and $x_i = 0$ if $i \notin D$. 
    A technical remark: it might be the case that there are more than one transitions in $\Delta$ that collapse to the same transition in $\Delta^D$, say $\delta_1 = (p, a, \ubf, q)$ and $\delta_2 = (p, a, \vbf, q)$ appear in $\Delta$ and collapse to $(p, a, q)$ in $\Delta^D$. If both transitions, $\delta_1$ and $\delta_2$, are seen infinitely often, we need to take care that we also see both infinitely often when translating the run $r$ back. This is possible using a round-robin procedure.

    Now observe that $sr'$ is a run of $\Amc$ on $w\beta$ (recall that $s$ ends in $f$, and $r'$ starts in $f$). Furthermore, we have $\rho(sr') = \rho(s) + \rho(r') = (v_1 + x_1, \dots, v_d + x_d)$, where $v_i + x_i = \infty$ if $i \in D$, and $v_i + x_i = v_i$ if $i \notin D$ by the observations above. Hence $\rho(sr') \in C$. Finally, $\Tmc^D$ enforces that at least one accepting state in $\Bmc^D_f$ is seen infinitely often, hence $w\beta \in L_\omega(\Amc)$.
\end{proof}

Finally we show the implication $(3) \Rightarrow (1)$.
\begin{lemma}
    If $L$ is RPBA-recognizable, then $L \in \LPAReg$.
\end{lemma}
\begin{proof}
    Let $\Amc = (Q, \Sigma, q_0, \Delta, F, C)$ be an RPBA. The intuition is as follows. An RPBA just needs to verify the counters a single time. Hence, we can recognize the prefixes of infinite words $\alpha \in B_\omega(\Amc)$ that generate the accepting hit with a PA. Further checking that an accepting state is seen infinitely often can be done with a Büchi automaton.

    Fix $f \in F$ and let $\Amc_f = (Q, \Sigma, q_0, \Delta, \{f\}, C)$ be the PA that is syntactically equal to $\Amc$ with the only difference that $f$ is the only accepting state. Similarly, let $\Bmc_f = (Q, \Sigma, f, \{(p,a,q) \mid (p,a,\vbf, q) \in \Delta\}, F)$ be the Büchi automaton obtained from $\Amc$ by setting $f$ as the initial state and the forgetting the vector labels. 
    
    We claim that $B_\omega(\Amc) = \bigcup_{f \in F} L(\Amc_f) \cdot L_\omega(\Bmc_f)$.

\smallskip
    $\Rightarrow$ To show $B_\omega(\Amc) \subseteq \bigcup_{f \in F} L(\Amc_f) \cdot L_\omega(\Bmc_f)$, let $\alpha \in B_\omega(\Amc)$ with accepting run $r = r_1 r_2 r_3 \dots$ where $r_i = (p_{i-1}, \alpha_i, \vbf_i, p_i)$. Let $k$ be arbitrary such that there is an accepting hit in $r_k$ (such a $k$ exists by definition) and consider the prefix $\alpha[1,k]$. Obviously $r_1 \dots r_k$ is an accepting run of $\Amc_{p_k}$ on $\alpha[1,k]$.
    Furthermore, there are infinitely many $j$ such that $p_j \in F$ by definition. In particular, there are also infinitely many $j \geq k$ with this property. Let $r'_i = (p_{i-1}, \alpha_i, p_i)$ for all $i > k$.  
    Then $r'_{k+1} r'_{k+2} \dots$ is an accepting run of $\Bmc_{p_k}$ on $\alpha_{k+1} \alpha_{k+2}\dots$ (recall that $p_k$ is the initial state of $\Bmc_{p_k}$). Hence we have $\alpha[1,k] \in L(\Amc_{p_k})$ and $\alpha_{k+1} \alpha_{k+2} \dots \in L_\omega(\Bmc_{p_k})$.

    \smallskip
    $\Leftarrow$ To show $\bigcup_{f \in F} L(\Amc_f) \cdot L_\omega(\Bmc_f) \subseteq B_\omega(\Amc)$, let $w \in L(\Amc_f)$ and $\beta \in L_\omega(\Bmc_f)$ for some $f \in F$. We show $w\beta \in B_\omega(\Amc)$.
    Let $s = s_1 \dots s_n$ be an accepting run of $\Amc_f$ on $w$, which ends in the accepting state $f$ with $\rho(s) \in C$ by definition.
    Furthermore, let $r = r_1 r_2 r_3 \dots$ be an accepting run of $\Bmc^D_f$ on $\beta$ which starts in the accepting state $f$ by definition. It is now easily verified that $sr'$ with $r' = r'_1r'_2r'_3\dots$ where $r'_i = (p_{i-1}, \alpha_i, \vbf_i, p_i)$ (for an arbitrary $\vbf_i$ such that $r'_i \in \Delta)$ is an accepting run of $\Amc$ on $w\beta$, as there is an accepting hit in $s_n$, and the (infinitely many) visits of an accepting state in $r$ translate one-to-one, hence $w\beta \in B_\omega(\Amc)$.
\end{proof}

\section{Characterization of \LPAPA and \LRegPA}
In this section we give a characterization of $\LPAPA$ and a characterization of $\LRegPA$. Grobler et al.~\cite{infiniteOurs} have shown that $\LPAPA \subsetneq \LSPBA$, \ie, SPBA are too strong to capture this class. However, restrictions of SPBA are a good candidate to capture $\LPAPA$ as well as $\LRegPA$. In fact we show that it is sufficient to restrict the appearances of accepting states to capture $\LPAPA$, as specified by the first theorem of this section. Further restricting the vectors yields a model capturing $\LRegPA$, as specified in the second theorem of this section. Recall that the condensation of $\Amc$ is the DAG of strong components of the underlying graph of $\Amc$. 

\begin{theorem}
\label{thm:lpapa}
    The following are equivalent for all $\omega$-languages $L \subseteq \Sigma^\omega$.
    \begin{enumerate}
        \item $L$ is of the form $\bigcup_i U_i V_i^\omega$, where $U_i, V_i \subseteq \Sigma^*$ are Parikh-recognizable.
        \item $L$ is recognized by an SPBA $\Amc$ with the property that accepting states appear only in the leaves of the condensation of $\Amc$, and there is at most one accepting state per leaf.
    \end{enumerate}
\end{theorem}

In fact, the proof of Grobler et al.~\cite{infiniteOurs} showing $\LPAPA \subseteq \LSPBA$ is constructive and \emph{almost} yields an SPBA with the desired property. A key notion are \emph{normalized} PA (on finite words), where a PA $\Amc = (Q, \Sigma, q_0, \Delta, F)$ is normalized if $F = \{f\}$ and $f$ has no outgoing transitions. It was shown that one can, given a PA $\Amc$, construct a normalized PA $\Amc_N$ with $L(\Amc_N) = L(\Amc) \setminus \{\varepsilon\}$.
For our proofs it is convenient to introduce a similar, yet stronger notion.
\begin{lemma}\label{lem:AIO}
Let $\Amc = (Q, \Sigma, q_0, \Delta, F, C)$ be a PA of dimension $d$. Then there exists a~PA~$\Amc^{IO}$ of dimension $d + 1$ with the following properties.
\begin{itemize}
    \item The initial state of $\Amc^{IO}$ is the only accepting state.
    \item $L(\Amc) \setminus \{\varepsilon\} = L(\Amc^{IO}) \setminus \{\varepsilon\}$.
    \item $SCC(\Amc) = \{Q\}$.
\end{itemize}
We say that $\Amc^{IO}$ is \emph{IO-normalized}.
\end{lemma}
\begin{proof}
    Define $\Amc^{IO} = \{Q \cup \{q_0'\}, \Sigma, q_0', \Delta^{IO}, \{q_0'\}, C \cdot \{1\})$, where 
    \begin{align*}
    \Delta^{IO} =&\ \{(p, a, \vbf \cdot 0, q) \mid (p, a, \vbf, q) \in \Delta\} \\
    \cup&\ \{(q_0', a, \vbf \cdot 0, q) \mid (q_0, a, \vbf, q) \in \Delta\} \\ \cup&\ \{(p, a, \vbf \cdot 1, q_0') \mid (p, a, \vbf, f) \in \Delta, f \in F\} \\
    \cup&\ \{(q_0', a, \vbf \cdot 1, q_0') \mid (q_0, a, \vbf, f) \in \Delta, f \in F\}.
    \end{align*}
    
    That is, $\Amc^{IO}$ is obtained from $\Amc$ by adding a fresh state $q_0'$, which is the initial state and only accepting state, inherits all outgoing transitions from $q_0$ and all in-going transitions from the accepting states. Furthermore, all transitions get a new counter, which is set to 0 except for the new ingoing transitions of $q_0'$ where the counter is set to $1$, and all vectors in $C$ are concatenated with $1$. 
    Finally, we remove all states that cannot reach $q'_0$ (such states can appear when shortcutting the ingoing transitions of $F$, and are useless in the sense that their removal does not change the accepted language; however, this removal is necessary for the third property). 
    We claim that $L(\Amc) \setminus \{\varepsilon\} = L(\Amc^{IO}) \setminus \{\varepsilon\}$.

    \smallskip
    $\Rightarrow$ To show $L(\Amc) \setminus \{\varepsilon\} \subseteq L(\Amc^{IO}) \setminus \{\varepsilon\}$, let $w_1 \dots w_n \in L(\Amc)$ for $n \geq 1$ with accepting run $r = r_1 \dots r_n$ where $r_i = (p_{i-1}, w_i, \vbf_i, p_i)$.
    By definition of $\Delta^{IO}$, there is a transition $r'_1 = (q_0', w_1, \vbf_1 \cdot 0, p_1)$ as well as a transition $r'_n = (p_{n-1}, w_n, \vbf_n \cdot 1, q_0')$ (or in case $n = 1$ the loop $(q_0', w_1, \vbf_1 \cdot 1, q_0')$). For all $1 < i < n$ define $r'_i = (p_{i-1}, w_i, \vbf_i \cdot 0, p_i)$. It is now easily verified that $r'_1 \dots r'_n$ (or simply $(q_0', w_1, \vbf_1 \cdot 1, q_0')$) is an accepting run of $\Amc^{IO}$ on $w_1 \dots w_n$.

    \smallskip
    $\Leftarrow$ To show $L(\Amc^{IO}) \setminus \{\varepsilon\} \subseteq L(\Amc) \setminus \{\varepsilon\}$, let $w_1 \dots w_n \in L(\Amc^{IO})$ for $n \geq 1$ with accepting run $r' = r'_1 \dots r'_n$ where $r'_i = (p_{i-1}, w_i, \vbf_i \cdot c_i, p_i)$ with $c_i \in \{0,1\}$.
    Observe that $p_0 = p_n = q_0'$, and for all $0 < i < n$ we have $p_i \neq q_0'$ as enforced by the additional counter (that is, $c_n = 1$ and $c_i = 0$ for all $i < n$, as $C \cdot \{1\}$ is the semi-linear set of~$\Amc^{IO}$).
    By definition of $\Delta^{IO}$ there is a transition $r_1 = (q_0, w_1, \vbf_1, p_1)$, and a transition $r_n = (p_{n-1}, w_n, \vbf_n, f)$ for some $f \in F$ in $\Delta$ (or in case $n = 1$ the transition $(q_0, w_1, \vbf_1, f)$). For all $1 < i < n$ define $r_i = (p_{i-1}, w_i, \vbf_i, p_i)$. It is now easily verified that $r_1 \dots r_n$ (or simply $(q_0, w_1, \vbf_1, f)$) is an accepting run of $\Amc$ on $w_1 \dots w_n$.
\end{proof}
Observe that $L(\Amc)^\omega = L(\Amc^{IO})^\omega$ for every PA $\Amc$, as $L^\omega = (L \setminus \{\varepsilon\})^\omega$ for every language~$L$ by definition. In fact, it is easily observed that we even have $S_\omega(\Amc^{IO}) = L(\Amc)^\omega$. We are now ready to proof the main theorem.

\begin{proof}[Proof of \Cref{thm:lpapa}]
$(1) \Rightarrow (2)$. Let $\Amc_i = (Q_i, \Sigma, q_i, \Delta_i, F_i)$ for $i \in \{1,2\}$ be PA and let $L = L(\Amc_1) \cdot L(\Amc_2)^\omega$. By \Cref{lem:AIO} and the observation above we can equivalently write $L = L(\Amc_1) \cdot S_\omega(\Amc^{IO}_2)$. 
As $\Amc^{IO}_2$ is IO-normalized it satisfies the property of the theorem.

We can now easily adapt the construction in \cite{infiniteOurs} showing that the concatenation of a Parikh-recognizable language and an SPBA-recognizable $\omega$-language is SPBA-recognizable to obtain an SPBA for $L(\Amc_1) \cdot S_\omega(\Amc^{IO}_2)$ that only keeps the accepting state of $\Amc^{IO}_2$, maintaining the property of the theorem. Finally, the closure under union is shown using a standard construction, hence combining SPBA with the desired property still yields an SPBA with the property. 
Overall, we obtain an SPBA $\Amc$ recognizing $L$ where the only accepting states appear in the leaves of $C(\Amc)$. \hfill$\lrcorner$

\medskip
$(2) \Rightarrow (1)$. Let $\Amc = (Q, \Sigma, q_0, \Delta, F, C)$ be an SPBA of dimension $d$ with the property of the theorem. 
Let $f \in F$ and let $\Amc_f = (Q, \Sigma, q_0, \Delta_f, \{f\}, C \cdot \{1\})$ with $\Delta_{f} = \{p,a,\vbf \cdot 0,q) \mid (p,a,\vbf, q) \in \Delta, q \neq f\} \cup \{(p, a, \vbf \cdot 1, f) \mid (p, a, \vbf, f) \in \Delta\}$ be the PA of dimension $d+1$ obtained from $\Amc$ by setting $f$ as the only accepting state with an additional counter that is 0 at every transition except of the ingoing transitions of $f$, where the counter is set to 1. 
Additionally all vectors in $C$ are concatenated with $1$. Similarly, let $\Amc_{f,f} = (Q, \Sigma, f, \Delta_{f}, \{f\}, C \cdot \{1\})$ be the PA of dimension $d+1$ obtained from $\Amc$ by setting $f$ as the initial state and only accepting state, where $\Delta_f$ is defined as for $\Amc_f$. We claim $S_\omega(\Amc) = \bigcup_{f \in F} L(\Amc_{f}) \cdot L(\Amc_{f,f})^\omega$.

\smallskip
$\Rightarrow$ To show $S_\omega(\Amc) \subseteq \bigcup_{f \in F} L(\Amc_{f}) \cdot L(\Amc_{f,f})^\omega$, let $\alpha \in S_\omega(\Amc)$ with accepting run $r = r_1 r_2 r_3 \dots$ where $r_i = (p_{i-1}, \alpha_i, \vbf_i, p_i)$. Let $k_1 < k_2 < \dots$ be the positions of accepting states in $r$, \ie, $p_{k_i} \in F$ for all $i \geq 1$.
First observe that the property in the theorem implies $p_{k_i} = p_{k_j}$ for all $i, j \geq 1$, \ie, no two distinct accepting states appear in~$r$, since accepting states appear only in different leaves of the condensation of $\Amc$. 
%However, it is not possible to reach $p_{k_i}$ from $p_{k_j}$ (or vice versa) by the definition of SCC and leaf. This is a contradiction, hence $p_{k_i} = p_{k_j}$ for all $i,j \geq 1$.

For all $j \geq 1$ define $r'_j = (p_{j-1}, \alpha_j, \vbf_j \cdot 0, p_j)$ if $j \neq k_i$ for all $i \geq 1$, and $r'_j = (p_{j-1}, \alpha_j, \vbf_j \cdot 1, p_j)$ if $j = k_i$ for some $i \geq 1$, \ie, we replace every transition $r_j$ by the corresponding transition in $\Delta_f$.

Now consider the partial run $r_1 \dots r_{k_1}$ and observe that $p_i \neq p_{k_1}$ for all $i < k_1$, and $\rho(r_1 \dots r_{k_1}) \in C$ by the definition of SPBA. Hence $r' = r'_1 \dots r'_{k_1}$ is an accepting run of~$\Amc_{p_{k_1}}$ on $\alpha[1, k_1]$, as only a single accepting state appears in $r'$, the newly introduced counter has a value of $1$ when entering $p_{k_1}$, \ie, $\rho(r') \in C \cdot \{1\}$, hence $\alpha[1, k_1] \in L(\Amc_{p_{k_1}})$.

Finally, we show that $\alpha[k_i + 1, k_{i+1}] \in L(\Amc_{p_{k_1},p_{k_1}})$.
Observe that $r'_{k_i + 1} \dots r'_{k_{i+1}}$ is an accepting run of $\Amc_{p_{k_1},p_{k_1}}$ on $\alpha[k_1 + 1, k_{i+1}]$: we have $\rho(r_{k_i + 1} \dots r_{k_{i+1}}) = \vbf \in C$ by definition. Again, as only a single accepting state appears in $r'_{k_i + 1} \dots r_{k_{i+1}}$, we have $\rho(r'_{k_i + 1} \dots r_{k_{i+1}}) = \vbf \cdot 1 \in C \cdot \{1\}$, and hence $\alpha[k_1 + 1, k_{i+1}] \in L(\Amc_{p_{k_1},p_{k_1}})$. We conclude $\alpha \in L(\Amc_{p_{k_1}}) \cdot L(\Amc_{p_{k_1}, p_{k_1}})^\omega$.

\smallskip
$\Leftarrow$ To show $\bigcup_{f \in F} L(\Amc_{f}) \cdot L(\Amc_{f,f})^\omega \subseteq S_\omega(\Amc)$, let $u \in L(\Amc_{f})$, and $v_1, v_2, \dots \in L(\Amc_{f,f})$ for some $f \in F$. We show that $uv_1v_2 \dots \in S_\omega(\Amc)$.

First let $u = u_1 \dots u_n$ and $r' = r'_1 \dots r'_n$ with $r'_i = (p_{i-1}, u_i, \vbf_i \cdot c_i, p_i)$, where $c_i \in \{0,1\}$, be an accepting run of $\Amc_{f}$ on $u$. 
Observe that $\rho(r') \in C \cdot \{1\}$, hence $\sum_{i \leq n} c_i = 1$, \ie,~$p_n$ is the only occurrence of an accepting state in $r'$ (if there was another, say $p_j$, then $c_j = 1$ by the choice of $\Delta_f$, hence $\sum_{i \leq n} c_i > 1$, a contradiction).
For all  $1 \leq i \leq n$ let $r_i = (p_{i-1}, u_i, \vbf_i, p_i)$. Then $r_1 \dots r_n$ is a partial run of $\Amc$ on $w$ with $\rho(r_1 \dots r_n) \in C$ and $p_n = f$.

Similarly, no run of $\Amc_{f,f}$ on any $v_i$ visits an accepting state before reading the last symbol, hence we continue the run from $r_n$ on $v_1, v_2, \dots$ using the same argument. Hence $uv_1v_2 \dots \in S_\omega(\Amc)$, concluding the proof.
\end{proof}

As a side product of the proof of \Cref{thm:lpapa} we get the following corollary, which is in general not true for SPBA.
\begin{corollary}
Let $\Amc = (Q, \Sigma, q_0, \Delta, F, C)$ be an SPBA with the property that accepting states appear only in the leaves of the condensation of $\Amc$, and there is at most one accepting state per leaf. Then we have $S_\omega(\Amc) = \bigcup_{f \in F} S_\omega(Q, \Sigma, q_0, \Delta, \{f\}, C)$.
\end{corollary}

By even further restricting the power of SPBA, we get the following characterization of $\LRegPA$.

\begin{theorem}
\label{thm:lregpa}
    The following are equivalent for all $\omega$-languages $L \subseteq \Sigma^\omega$.
    \begin{enumerate}
        \item $L$ is of the form $\bigcup_i U_i V_i^\omega$, where $U_i \subseteq \Sigma^*$ is regular and $V_i \subseteq \Sigma^*$ is Parikh-recognizable.
        \item $L$ is recognized by an SPBA $\Amc$ with the following properties.
        \begin{enumerate}
            \item At most one state $q$ per leaf of the condensation of $\Amc$ may have ingoing transitions from outside the leaf, this state $q$ is the only accepting state in the leaf, and there are no accepting states in non-leaves.
            \item only transitions connecting states in a leaf may be labeled with a non-zero vector. 
        \end{enumerate}
    \end{enumerate}
\end{theorem}
Observe that property (a) is a stronger property than the one of \Cref{thm:lpapa}, hence, SPBA with this restriction are at most as powerful as those that characterize $\LPAPA$. However, as a side product of the proof we get that property (a) is equivalent to the property of \Cref{thm:lpapa}. Hence, property (b) is mandatory to sufficiently weaken SPBA such that they capture $\LRegPA$. 
In fact, using the notion of IO-normalization, we can re-use most of the ideas in the proof of \Cref{thm:lpapa}.

\begin{proof}[Proof of \Cref{thm:lregpa}]

$(1) \Rightarrow (2)$. We can trivially convert an NFA into an equivalent PA by labeling every transition with $0$ and choosing $C = \{0\}$.
Let $\Amc$ be an arbitrary PA and observe that $\Amc^{IO}$ has only a single SCC by definition. 
Again, we have $L(\Amc)^\omega = S_\omega(\Amc^{IO})$ and the constructions for concatenation and union do not destroy the properties, hence we obtain an SPBA of the desired form. \hfill $\lrcorner$

\medskip
$(2) \Rightarrow (1)$ Let $\Amc = (Q, \Sigma, q_0, \Delta, F, C)$ be an SPBA of dimension $d$ with properties~(a) and (b). Fix $f \in F$ and let \mbox{$\Bmc_f = (Q_f, \Sigma, q_0, \{(p, a, q) \mid (p, a, \vbf, q) \in \Delta, p,q \in Q_f\}, \{f\})$} with $Q_f = \{q \in Q \mid q \text{ appears in a non-leaf SCC of } C(\Amc)\} \cup \{f\}$
be the NFA obtained from $\Amc$ by removing all leaf states except $f$, and removing all labels from the transitions.
Recycling the automaton from \Cref{thm:lpapa}, let \mbox{$\Amc_{f,f} = (Q, \Sigma, f, \Delta_{f}, \{f\}, C \cdot \{1\})$} with $\Delta_{f} = \{p,a,\vbf \cdot 0,q) \mid (p,a,\vbf, q) \in \Delta, q \neq f\} \cup \{(p, a, \vbf \cdot 1, f) \mid (p, a, \vbf, f) \in \Delta\}$.
We claim $S_\omega(\Amc) = \bigcup_{f \in F} L(\Bmc_f) \cdot L(\Amc_f)^\omega$.

\smallskip
$\Rightarrow$ To show $S_\omega(\Amc) = \bigcup_{f \in F} L(\Bmc_f) \cdot L(\Amc_{f,f})^\omega$, let $\alpha \in S_\omega(\Amc)$ with accepting run $r = r_1r_2r_3 \dots$ where $r_i = (p_{i-1}, \alpha_i, \vbf_i, p_i)$, and let $k_1< k_2< \dots$ be the positions of the accepting states in $r$, and consider the partial run $r_1 \dots r_{k_1}$ (if $k_1 = 0$, \ie, the initial state is already accepting, then $r_1 \dots r_{k_1}$ is empty). 

By property (a) we have that $p_{k_1}$ is the first state visited in $r$ that is located in a leaf of $C(\Amc)$. Hence $r'_1 \dots r'_{k_1}$, where $r'_i = (p_{i-1}, \alpha_i, p_i)$, is an accepting run of $\Bmc_{p_{k_1}}$ on $\alpha[1, k_1]$ (in the case $k_1 = 0$ we define $\alpha[1, k_1] = \varepsilon$).

By the same argument as in the proof of \Cref{thm:lpapa} we have $p_{k_i} = p_{k_j}$ for all $i,j \geq 1$, hence $\alpha[k_i + 1, k_{i+1}] \in L(\Amc_{p_{k_1}, p_{k_1}})$, and hence $\alpha \in L(\Bmc_{p_k}) \cdot L(\Amc_{p_{k_1}, p_{k_1}})^\omega$.

\smallskip
$\Leftarrow$ To show $\bigcup_{f \in F} L(\Amc_{f}) \cdot L(\Amc_{f,f})^\omega \subseteq S_\omega(\Amc)$, let $u \in L(\Bmc_{f})$, and $v_1, v_2, \dots \in L(\Amc_{f,f})$ for some $f \in F$. We show that $uv_1v_2 \dots \in S_\omega(\Amc)$.

First observe that properties (a) and (b) enforce that $\0 \in C$, as the accepting state of a leaf of $C(\Amc)$ is visited before a transition labeled with a non-zero can be used.
Let $u = u_1 \dots u_n$ and $s_1 \dots s_n$ with $s_i = (p_{i_1}, u_i, p_i)$ be an accepting run of $\Bmc_f$ on $u$. Define $s'_i = (p_{i_1}, u_i, \0, p_i)$ and observe that $s'_1 \dots s'_n$ is a partial run of $\Amc$ with $\rho(s'_1 \dots s'_n) \in C$ and $p_n = f$ by the observation above.

Again we can very similarly continue the run on $v_1, v_2, \dots$ using the same argument. Hence $uv_1v_2 \dots \in S_\omega(\Amc)$, concluding the proof.
\end{proof}

\section{Conclusion}
%Sometimes you just need a few minutes on the bike to solve a problem you were thinking of more than three years. Then you can rush another paper within two weeks.
%However I do not necessarily recommend this, as this can be super exhausting.
We conclude with an overview of our results shown in \Cref{fig:overview}.
\begin{figure}
\centering
    \begin{tikzpicture}[%
      node distance=27mm,>=Latex,
      initial text="", initial where=below left,
      every state/.style={rectangle,rounded corners,draw=black,thin,fill=black!5,inner sep=1mm,minimum size=6mm},
      every edge/.style={draw=black,thin}
    ]
    \node[state] (cRPA) {complete RPA};
    \node[state,right = 1cm of cRPA] (RPA) {RPA};
    \node[state,above right = 1cm and 0cm of RPA]  (RPBA) {RPBA = LPBA = $\LPAReg$};
    \node[state,below right = 1cm and -1cm of cRPA]  (reg) {BA = $\LRegReg$};
    \node[state,right = 1cm of reg] (regPA) {SPBA $(**)$ = $\LRegPA$};
    \node[state,below right = 1cm and 0cm of RPBA] (PPBA) {PPBA};
    \node[state,right = 1cm of PPBA] (PAPA) {SPBA $(*)$ = $\LPAPA$};
    \node[state,below right = 1cm and -1.15cm of PAPA] (SPBA) {SPBA};

    \path[-{Latex}]
    (cRPA) edge node[above=3mm] {\small\Cref{remark:RPA}} (RPA)
    (RPA) edge node[below right] {\small\Cref{lem:RPARPBA}} (RPBA)
    (RPBA) edge node[above right = -1mm and -1mm] {\small\Cref{cor:LRPBAtoLPPBA}} (PPBA)
    (PPBA) edge node[above=2.5mm] {\small \cite{blindcounter, infiniteOurs}, \cite{infiniteZimmermann}} (PAPA)
    (PAPA) edge node[right] {\small\cite{infiniteOurs}} (SPBA)
    (reg) edge (regPA)
    (regPA) edge[bend right = 15] (PAPA);
    ;

    \path[dotted]
    (cRPA) edge node {$\neq$} node[left=3mm] {\small\Cref{lem:RPAregular}} (reg)
    (RPA) edge node {$\neq$} (reg)
    (RPA) edge node {$\neq$} (regPA)
    (PPBA) edge node {$\neq$} (regPA)
    ;

    \node[above=0mm of RPBA] {\small \Cref{thm:LimitEqualsReach}};
    \node[below=0mm of regPA] {\small \Cref{thm:lregpa}};
    \node[below=0mm of reg] {\small \Cref{thm:buechi}};
    \node[above=0mm of PAPA] {\small \Cref{thm:lpapa}};

    \draw[-{Latex}] (reg) -- ++(-3, 0) -- ++(0,3.2) -- (RPBA);

    \node[align=left, anchor=west] at (-1.5, -3) {\small $(*)$ At most one state $q$ per leaf of $C(\Amc)$ may have ingoing transitions from outside the  leaf, this \\\small\phantom{$(*)$} state $q$ is the only accepting state in the leaf, and there are no accepting states in non-leaves;};
    \node[align=left, anchor=west] at (-1.65, -3.66) {\small $(**)$ and only transitions connecting states in leaves may be labeled with non-zero vectors.};

    \end{tikzpicture}        
    \caption{Overview of our results. Arrows mean strict inclusions and $\neq$ means orthogonal.}
    \label{fig:overview}
\end{figure}
To finalize the picture, we observe the following.
\begin{observation}
\mbox{}
\begin{enumerate}
    \item There are RPA-recognizable $\omega$-languages not contained in $\LRegPA$, for example $\{a^n b^n \mid n \geq 0\} \{a\}^\omega$.
    \item We have $\LRegPA \subseteq \LPAPA$ by definition, and there are $\omega$-languages contained in $\LPAPA$ that are not contained in $\LRegPA$, for example $\{a^n b^n \mid n \geq 1\} \{c^n d^n \mid n \geq 1\}^\omega$.
    \item We have $\LRegReg \subseteq \LRegPA$ by definition, and there are $\omega$-languages contained in $\LRegPA$ that are neither PPBA-recognizable nor $\omega$-regular, as witnessed by many $\omega$-closures of Parikh-recognizable languages (which are trivially contained in $\LRegPA$), for example $\{a^n b^n \mid n \geq 1\}^\omega$ (a formal proof can be found for blind counter machines in \cite{blindcounter}, which are known to be equivalent to PPBA \cite{infiniteOurs}).
\end{enumerate}
\end{observation}

Finally, we recall that deterministic $\omega$-regular languages are characterized as regular arrow-languages $\vec{L}$, where $\vec{L} = \{\alpha \mid \alpha[1,i] \in L \text{ for infinitely many }i\}$ \cite{thomasinfinite}. This characterization can easily be adapted to show that deterministic PPBA-recognizable $\omega$-languages are captured by arrows of deterministic Parikh-recognizable languages. We conjecture that PPBA-recognizable $\omega$-languages are captured by arrows of Parikh-recognizable languages yielding a similar characterization of $\LPPBA$.

\newpage
\bibliographystyle{plain}
\bibliography{lit}

\begin{thebibliography}{10}

\bibitem{allredCounting}
Jo\"el Allred and Ulrich Ultes-Nitsche.
\newblock $k$-counting automata.
\newblock {\em RAIRO - Theoretical Informatics and Applications - Informatique
  Th\'eorique et Applications}, 46(4):461--478, 2012.

\bibitem{modelchecking}
Christel Baier and Joost-Pieter Katoen.
\newblock {\em Principles of Model Checking}.
\newblock The MIT Press, 2008.

\bibitem{bojanczykbeyond}
Mikolaj Bojanczyk.
\newblock {Beyond omega-Regular Languages}.
\newblock In Jean-Yves Marion and Thomas Schwentick, editors, {\em 27th
  International Symposium on Theoretical Aspects of Computer Science}, volume~5
  of {\em Leibniz International Proceedings in Informatics (LIPIcs)}, pages
  11--16, Dagstuhl, Germany, 2010. Schloss Dagstuhl--Leibniz-Zentrum fuer
  Informatik.

\bibitem{buechi}
J.~Richard Büchi.
\newblock Weak second-order arithmetic and finite automata.
\newblock {\em Mathematical Logic Quarterly}, 6(1‐6):66--92, 1960.

\bibitem{blindcounter}
Henning Fernau and Ralf Stiebe.
\newblock Blind counter automata on omega-words.
\newblock {\em Fundam. Inform.}, 83:51--64, 2008.

\bibitem{infiniteOurs}
Mario Grobler, Leif Sabellek, and Sebastian Siebertz.
\newblock Parikh automata on infinite words, 2023.

\bibitem{infiniteZimmermann}
Shibashis Guha, Ismaël Jecker, Karoliina Lehtinen, and Martin Zimmermann.
\newblock Parikh automata over infinite words, 2022.

\bibitem{klaedtkeruess}
Felix Klaedtke and Harald Rue{\ss}.
\newblock Monadic second-order logics with cardinalities.
\newblock In Jos C.~M. Baeten, Jan~Karel Lenstra, Joachim Parrow, and
  Gerhard~J. Woeginger, editors, {\em Automata, Languages and Programming},
  pages 681--696, Berlin, Heidelberg, 2003. Springer.

\bibitem{beyondomegabs}
Dario~Della Monica, Angelo Montanari, and Pietro Sala.
\newblock Beyond $\omega${BS}-regular languages: $\omega$t-regular expressions
  and counter-check automata.
\newblock {\em Electronic Proceedings in Theoretical Computer Science},
  256:223--237, 2017.

\bibitem{thomasinfinite}
Wolfgang Thomas.
\newblock {\em Automata on Infinite Objects}, page 133–191.
\newblock MIT Press, Cambridge, MA, USA, 1991.

\end{thebibliography}

\end{document}